\documentclass[journal,onecolumn,12pt]{IEEEtran}
\usepackage{eulervm,soul}
\usepackage[colorlinks,citecolor=blue,linkcolor=blue,urlcolor=blue]{hyperref}
\usepackage[utf8]{inputenc} 
\usepackage[T1]{fontenc}
\usepackage{ifthen}
\usepackage{cite}
\usepackage[cmex10]{amsmath} 

\usepackage[final]{graphicx}

\usepackage{amsthm}
\usepackage{amssymb}
\usepackage{euscript}
\usepackage{subfigure}
\usepackage[colorinlistoftodos]{todonotes}
\usepackage{bbm}
\usepackage{flushend}
\usepackage{mathtools}
\usepackage{setspace}
\usepackage{soul}

\newtheorem{theorem}{Theorem}
\newtheorem{definition}{Definition}
\newtheorem{lemma}{Lemma}

\newtheorem{corollary}{Corollary}
\begin{document}
\title{Balanced Nonadaptive Redundancy Scheduling}
\author{Amir Behrouzi-Far,
        and~Emina Soljanin,~\IEEEmembership{IEEE Fellow}%
\thanks{This work was supported by NSF Award under Grant CIF-1717314 and CCF-1559855. This work was presented in part in the 2019 57'th Annual Allerton Conference on Communication, Control and Computing \cite{behrouzi2019redundancy} and 2019 IEEE international conference on Big Data \cite{behrouzi2019scheduling}.

The authors are with the Department of Electrical and Computer Engineering, Rutgers University, New Brunswick, NJ 08901 USA (e-mail: amir.behrouzifar@rutgers.edu; emina.soljanin@rutgers.edu).}
}

\maketitle

\begin{abstract}
Distributed computing systems implement redundancy to reduce the job completion time and variability. Despite a large body of work about computing redundancy, the analytical performance evaluation of redundancy techniques in queuing systems is still an open problem. In this work, we take one step forward to analyze the performance of scheduling policies in systems with redundancy. In particular, we study the pattern of shared servers among replicas of different jobs. To this end, we employ combinatorics and graph theory and define and derive performance indicators using the statistics of the overlaps. We consider two classical nonadaptive scheduling policies: random and round-robin. We then propose a scheduling policy based on combinatorial block designs. Compared with conventional scheduling, the proposed scheduling improves the performance indicators. We study the expansion property of the graphs associated with round-robin and block design-based policies. It turns out the superior performance of the block design-based policy results from better expansion properties of its associated graph. As indicated by the performance indicators, the simulation results show that the block design-based policy outperforms random and round-robin scheduling in different scenarios. Specifically, it reduces the average waiting time in the queue to up to $25\%$ compared to the random policy and up to $100\%$ compared to the round-robin policy. 
\end{abstract}


\newpage
\section{Introduction}
Redundancy in queuing systems could reduce the average sojourn time of jobs \cite{dean2013tail}, \cite{joshi2014delay}. It is widely studied in theory \cite{joshi2017efficient}, \cite{lee2017speeding}, and employed in practice, \cite{he2010comet}, \cite{bernardin2006using}. Redundancy is particularly beneficial in systems with high variability in job service times \cite{gardner2017redundancy}.
In a system with multiple servers, redundancy could reduce the average time a job spends in the system. Particularly, when the time required to complete a job is variable across the jobs, redundancy can improve the performance by \textit{diversifying} the queuing time of jobs. That is, by concurrent execution at multiple servers, a job's queuing time is the minimum queuing time across its copies \cite{behrouzi2018effect}.

Scheduling research in queuing systems has long history \cite{wang1985load}. It has had different objectives, e.g., load balancing \cite{harchol1996exploiting}, fairness \cite{bansal2001analysis}, minimizing the average/variability of job lifetime in the system \cite{wierman2005classifying}, maximizing resource utilization \cite{peng2015r}. As these objectives may not align, various other works consider the trade-off between different objectives \cite{bansal2001analysis}. In general, scheduling policies can be categorized into \textit{adaptive} and \textit{nonadaptive} policies.  Adaptive policies use the information of all or some queues' length/servers' workload to pick server(s) for arriving jobs. Join the shortest queue \cite{gupta2007analysis} policy uses the information about all servers. In contrast, the power-of-d choices \cite{gardner2017redundancy} policies use the information on a randomly picked small subset of servers. Nonadaptive policies, e.g., random or round-robin \cite{harchol1998choosing}, make blind but possibly structured decisions. Nonadaptive policies are used in practical systems when the estimate about the actual job's resource requirement, and thus the actual workload on the servers, is not accurate \cite{GoogleProvisioning}. Over-provisioning/under-provisioning of resources is common in such circumstances.

Scheduling problems in redundancy systems require selecting multiple servers for each job. Therefore, a scheduling problem could be much more complex in a system with redundancy than in no-redundancy systems. Classical scheduling policies, designed for no-redundancy systems, may not have the expected performance in scenarios with redundancy. That is because the goal of redundancy, which is to \textit{diversify} the sojourn time of jobs, has not been considered in the design of classical scheduling policies. Redundancy could increase the cost of running a job, and thus special attention has to be paid to fully utilize the extra resources used for redundant requests. In other words, one needs to ensure that redundancy brings diversity to the sojourn time of jobs \cite{peng2020diversityparallelism}. 

In this work, we study the performance of nonadaptive scheduling policies in queuing systems with redundancy. We consider an $(n,d,1)$ partial fork-join system \cite{rizk2016stochastic}, with early cancellation of redundant copies. Partial fork-join systems are extensively studied in distributed computing systems with redundancy (see \cite{behrouzi2021efficient} and references therein), and they are widely employed in commercial computing frameworks, such as Kubernetes \cite{burns2016borg}. Early cancellation of redundant copies of a job is shown to be optimal when the tail distribution of jobs' service time is log-concave \cite{joshi2017efficient}. Our goal is to understand what determines a scheduling policy's performance. 

Our study focuses on the statistics of overlaps between servers assigned to different jobs under a given scheduling policy. Diversification can be improved by a careful design of the \textit{overlaps} between the servers assigned to different jobs. The role of overlaps in the performance of distributed systems with redundancy has been studied in the literature: \cite{behrouzi2021efficient} examined how controlling overlaps between job batches affects job execution time, and \cite{aktacs2021evaluating} observed that an overlap-aware redundant data allocation can improve the performance in a distributed storage model.

To quantify overlaps, we introduce three performance indicators, and judge and compare the performance of scheduling policies based on these indicators. To that end, we develop an analogy to some occupancy (urns \& balls) problems. Urns \& balls models refer to basic probabilistic experiments in which
balls are thrown randomly into urns. We are, in general, interested in various patterns of urn occupancy (e.g., the number of empty urns). These models are central in many disciplines such as combinatorics, statistics, analysis of algorithms, and statistical physics. We were able to use the rich literature on the subject. In addition, our scheduling questions generated new fundamental problems in the field \cite{behrouzi2019average}, which have been subsequently addressed in \cite{michelen2019short}.

We next develop a graph-theoretical model for nonadaptive scheduling and use it to explore the hidden properties of the round-robin and block design-based policies. It turns out that the expansion property of a scheduling policy's associated graph is an indicator of its performance in queuing systems with redundancy. Finally, our simulation results show that our proposed scheduling policy achieves a lower average queuing time than the random and round-robin policies, as suggested by the performance indicators, by a considerable margin.

We first consider two canonical nonadaptive scheduling policies: random and round-robin.  We then propose a new nonadaptive scheduling policy based on a combinatorial block design. Block designs are combinatorial structures with many applications, including design of experiments, software testing, cryptography, and finite and algebraic geometry. They have recently been considered for straggler mitigation in distributed machine learning \cite{sakorikar2021variants,kadhe2019gradient,behrouzi2019redundancy}. The importance of such structures come from their ability to regulate overlaps between the copies of different jobs, as shown in \cite{behrouzi2019redundancy}.

This paper is organized as follows. In Section \ref{sec:model} we introduce the system model, the analogy to the urns and balls problem, and our performance indicators. Using the analogy to the urns and balls problem in Section \ref{sec:indicators} we derive performance indicators for random and round-robin scheduling policies. In Section \ref{sec:BIBD} we propose a new scheduling policy and derive its performance indicators. In Section \ref{sec:numerical_comp} we present comparisons of the performance indicators. We connect the performance indicators of scheduling policies with the expansion property of their associated graphs in Section \ref{sec:graph}. We present the simulation results of a queuing system in Section \ref{sec:simResult}. Finally, we offer the concluding remarks and future directions in Section \ref{sec:conclusion}.

\section{System Model and Performance Measures}\label{sec:model}
We first present the model of our system and give some examples. We then develop an analogy to some urns \& balls problems. Finally, we use this analogy to motivate and define our performance measures and indicators. 

\subsection{Queuing Model and Scheduling Policies} \label{subsec:scheduling_policies}
We consider a system with $n$ identical servers and a scheduler, as shown in Fig.~\ref{fig:sysModel}. 
    \begin{figure}[hbt]
        \centering
        \includegraphics[width=.4\columnwidth]{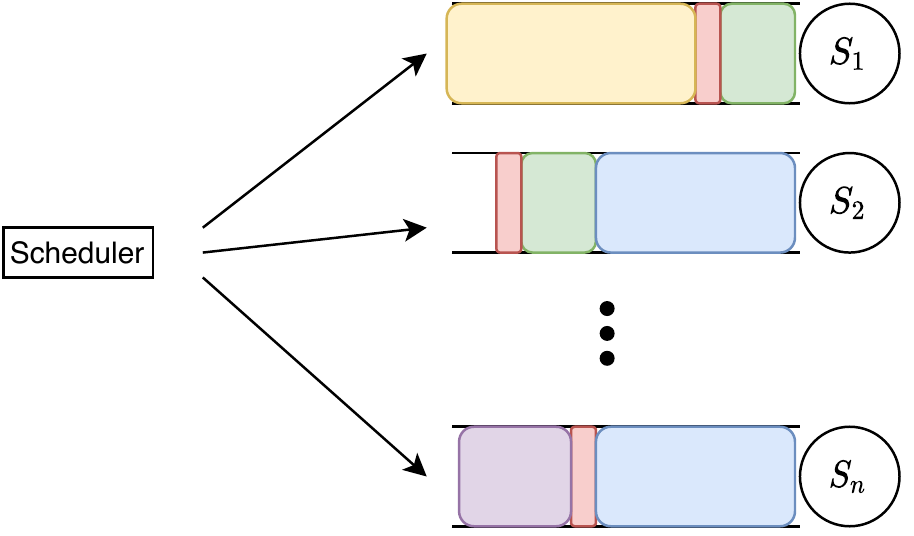}
        \caption{Scheduling model: Jobs get redundantly assigned to servers. The scheduling is nonadaptive, i.e., it has no information about the workload in the queues. Observe that copies of two different jobs may be assigned to the same server. These overlaps are unavoidable and depend on the scheduling policy.}
        \label{fig:sysModel}
    \end{figure}

The scheduler sends $d$ copies of each job to $d\le n$ servers. As soon as one of the copies starts the service, the redundant copies get instantaneously cancelled. Observe that copies of two different jobs may be assigned to the same server. These overlaps are unavoidable and depend on the scheduling policy. Furthermore, reducing a job's overlap with one group of jobs may increase its overlap with another group of jobs.

We focus on nonadaptive scheduling policies. Such policies make assignments of jobs to servers without taking into account the servers' workloads. Our goal is to understand and evaluate the effect of scheduling policies on the queuing time of jobs. The job's queuing time is the waiting time in the queue from the moment the job is scheduled until one of its copies gets into service. We consider two conventional nonadaptive scheduling policies: random and round-robin.  Under the random policy, the scheduler assigns $d$ copies of each job to $d$ randomly selected servers. Round-robin policy, on the other hand, takes a more structured approach. It assigns the $j$th copy of $i$th arriving job to the server indexed by $((i-1)d+j)\mod n$, for $j=1,2,\dots,r$. Furthermore, we propose a new nonadaptive scheduling policy based on combinatorial block designs. We will see that the block design-based policy,  can reduce the average queuing time of jobs by better managing the overlaps between the set of selected servers across the jobs. 

\subsection{Examples}\label{subsec:IntroExamples}
Fig.~\ref{fig:overlapping_queues} illustrates two scheduling policies with different patterns of overlapping queues. 
    \begin{figure}[hbt]
        \centering
        \includegraphics[width=.6\columnwidth]{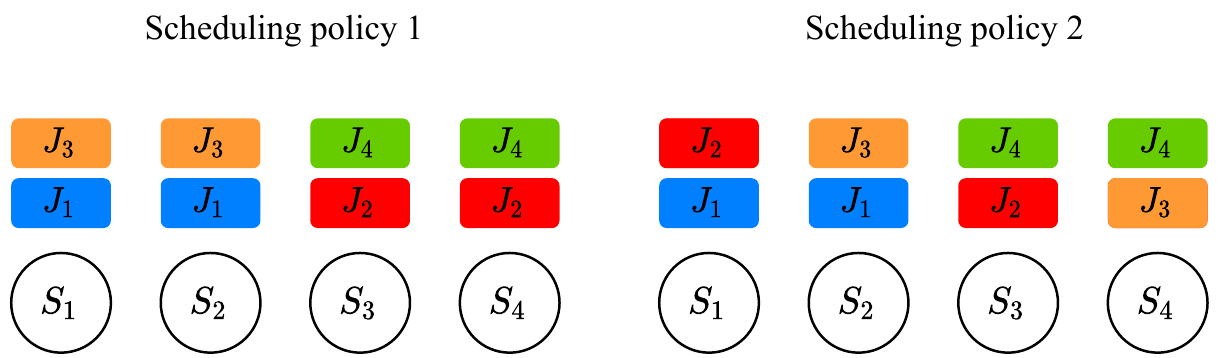}
        \caption{Scheduling with overlapping queues. Despite introducing redundancy, scheduling policy 1 provides no diversity on the jobs' queuing time. Scheduling policy 2, on the other hand, provides diversity. A scheduling policy's ability to provide diversity can be measured by studying the resulting overlaps between the replicas of jobs.}
        \label{fig:overlapping_queues}
    \end{figure}
The system is comprised of four servers $\{S_i\}_{i=1}^{4}$ and four jobs $\{J_k\}_{k=1}^{4}$. Each job is assigned to $d=2$ servers. With scheduling policy 1, each job shares its two servers with two replicas of another job. That is, scheduling policy 1 provides no diversity in overlaps. For instance, $J_1$ and $J_3$ share the same servers $S_1$ and $S_2$, making $J_3$ experience identical queuing time across the replicas. Therefore, despite redundancy, scheduling policy 1 provides no queuing diversity. For example, in the case of deterministic service time of jobs across the replicas, the redundancy offers no benefit under scheduling policy 1, as all replicas of a given job experience the same queuing time. On the other hand, with scheduling policy 2, each replica of a given job shares the server with a replica of a different job. That is, a job has overlapping queues with two different jobs and thus may experience diversity in queuing time across its replicas. As an example, $J_4$ shares $S_3$ with $J_2$ and it shares $S_4$ with $J_3$. Therefore, the queuing time of $J_2$ is the minimum of service time of $J_2$ and $J_3$. 
The resulting diversity reduces the average queuing time  and thus policy 2 is preferred over policy 1 in a queuing system with redundancy.

\subsection{The Analogy to An Urns \& Balls Problem}\label{subsec:perf_indicators}

Let us consider a queuing system with no servers. Jobs arrive, get scheduled redundantly, but do not enter the service, and thus there is no departure. We refer to this scenario as \textit{arrivals-only}. We can model this scenario as the urns and balls problem, where urns and balls are, respectively, equivalent to queues and jobs. With this connection, we can use the rich literature on such combinatorial problems to understand the overlapping pattern of scheduling policies in queuing systems with redundancy \cite{sachkov1996combinatorial}.

To quantify the statistics of overlaps for a given scheduling policy in the arrivals-only scenario, we proceed as follows. Suppose that in the urns and balls analogy, each ball is replicated to $d$ urns, chosen according to a ``scheduling'' policy. For a given number $T$ of balls scheduled into the urns (resulting in $dT$ replicas), there exist a set $\CMcal{P}=\{(i,j)|1\leq i,j\leq T, i\neq j\}$ of $\binom{T}{2}$ pairs of balls .

\begin{definition}\label{def:oX}
Let the random variable $X$ be the indices of a pair of balls chosen randomly from $\CMcal{P}$. We define the random number $o_X$ as the number of overlapping urns between the pairs of balls in $X$
\end{definition}

Our first performance indicator measures $\mathbbm{E}_{X}[o_X]$; the average number of overlapping urns between a random pair of balls. In general, a scheduling policy with smaller $\mathbbm{E}_{X}[o_X]$ provides higher average diversity in jobs' queuing time.

Consider the examples in Fig.~\ref{fig:overlapping_queues}. Scheduling policy 1 has $\mathbbm{E}_{X}[o_X]=2/3$, as there are six pairs of balls, two pairs of which overlap in 2 servers and the other four pairs overlap in no servers. The second performance indicator measures $\mathbbm{E}_{X}[o_X^2]$; the degree of ``evenness'' in the number of overlaps across the pairs of balls. A scheduling policy with smaller $\mathbbm{E}_{X}[o_X^2]$ provides more even overlaps between jobs' replicas. In Fig.~\ref{fig:overlapping_queues}, the scheduling policy 1 has $\mathbbm{E}_{X}[o_X^2]=4/3$. With scheduling policy 2, there are six pairs, four pairs of which have one overlap, and the other 2 has no overlap. Therefore, scheduling policy two has $\mathbbm{E}_{X}[o_X^2]=2/3$ and thus provides higher diversity.

\subsection{System Performance and its Indicators}

The performance metric we are concerned with is the average job queuing time. In systems with redundancy, analysis of the average queuing time is complex \cite{raaijmakers2018delta}. In this work, we introduce performance indicators for scheduling policies. We reason about the effectiveness of those indicators for understanding the relative performance of different scheduling policies in a queuing system with redundancy. 

Ideally the copies of an arriving job to a queuing system has to be scheduled to the servers with small workload. As a nonadaptive scheduling policy has no side information about the servers' workload, picking servers according to their workload is not possible. Therefore, the best a nonadaptive policy can do is to diversify the queuing time of a job by sending its copies to servers with varying workloads, hoping that at least one of the picked servers has a relatively small work left. For a given scheduling policy, we quantify its resulting diversity in jobs' queuing time by the number of overlapping queues a randomly selected job experiences with other jobs in the system.

%

We next introduce and define three performance indicators.
We consider a setting with $n$ urns and $T$ balls, where $d$ copies of each ball are placed in $d$ distinct urns. The $d$ urns are chosen according to a given scheduling policy. The assignment of $d$ copies of each ball may be referred to as \textit{$d$-redundant}. Let $N_i^T$, $i=1,2,\dots,n$ denote the numbers of balls in urn $i$ after placing all $T$ balls $d$-redundantly. Our performance indicators are defined using the statistics of random variables $o_X$ (cf. Definition~\ref{def:oX}) and $N_i$, after placing $T$ balls into $n$ urns $d$-redundantly, as $T\rightarrow\infty$.

\begin{definition}\label{def:LBF}
The \ul{Load Balancing Factor} (LBF) of a given scheduling policy is the ratio of the average number of balls in the minimum and maximum loaded urns,
    \begin{equation}
        \textup{LBF}_{\text{policy}}
        \coloneqq\lim_{T\to\infty}\frac{\mathbb{E}\left[\min\{N_i\}_{i=1}^{n}\right]}{\mathbb{E}\left[\max\{N_i\}_{i=1}^{n}\right]}.
    \end{equation}
\end{definition}

\begin{definition}\label{def:AOF}
The \ul{Average Overlap Factor} (AOF) of a given scheduling policy is defined as the reciprocal of the first moment of $o_X$,
    \begin{equation}\label{equ:AOF}
        \textup{AOF}_{policy}\coloneqq\lim_{T\to\infty}\frac{1}{\mathbb{E}[o_X]}.
    \end{equation}{}
\end{definition}{}

\begin{definition}\label{def:ODF}
The \ul{Overlap Diversity Factor} (ODF) of a given scheduling policy is defined as the reciprocal of the second moment of $o_X$,
    \begin{equation}\label{equ:ODF}
        \textup{ODF}_{policy}\coloneqq\lim_{T\to\infty}\frac{1}{\mathbb{E}[o_X^2]}.
    \end{equation}{}
\end{definition}{}

\section{Computing Performance Indicators}\label{sec:indicators}
We derive the closed-form expressions of the three performance indicators for the random and round-robin. We assume $n$ is a multiple of $T$. The extension of our results to general values of $T$ is straightforward. We numerically evaluate the derived expressions and observe that a large LBF, AOF, and/or ODF is an indication of a shorter average queuing time.

\subsection{Random Scheduling}
With the random policy, each ball is assigned to $d$ urns, chosen uniformly at random without replacement. With $d=1$, the problem boils down to the classical urns and balls problem where several analytical results exist for the occupancy of the urns, e.g., \cite{raab1998balls,dubhashi1998balls}. With $d>1$, \cite{michelen2019short}, inspired by \cite{behrouzi2019average}, derives asymptotic results for the number of balls in the maximum loaded bin, when $n\rightarrow\infty$. We adopt Lemma~\ref{lem:random_LBF}  from \cite{michelen2019short}.

\begin{lemma}\label{lem:random_LBF}
After throwing $T$ balls $d$-redundantly to $n$ urns, when urns are chosen uniformly at random for each ball, the LBF can be asymptotically approximated by,
    \begin{equation}
        \textup{LBF}_{random}\approx\textup{max}\left\{0,\frac{\frac{Td}{n}-\sqrt{\frac{2Td(n-d)\textup{log}(n)}{n^2}}}{\frac{Td}{n}+\sqrt{\frac{2Td(n-d)\textup{log}(n)}{n^2}}}\right\},
    \end{equation}{}
as $n\rightarrow\infty$.
\end{lemma}{}
\begin{proof}
The average number of balls in the maximum loaded urn, after throwing $T$ balls $d$-redundantly to $n$ urns with random scheduling, is shown in \cite{behrouzi2019average} that follows
    \begin{equation}
        \mathbb{E}[N_{n:n}]=\frac{Td}{n}+C_{n,d}\sqrt{T}+O\left(1/\sqrt{T}\right).
    \end{equation}{}
Then it is proved in \cite{michelen2019short} that,
    \begin{equation}
        C_{n,d}\approx\sqrt{\frac{2d(n-d)\textup{log}(n)}{n^2}},\quad\textup{as}\quad  n\rightarrow\infty.
    \end{equation}{}
Therefore,
    \begin{equation}
        \mathbb{E}[N_{n:n}]\approx\frac{Td}{n}+\sqrt{\frac{2Td(n-d)\textup{log}(n)}{n^2}},\quad\textup{as}\quad  n\rightarrow\infty.
    \end{equation}{}
From \cite{michelen2019short}, it is straightforward that the number of balls in the minimum loaded urn could be approximated by,
    \begin{equation}
        \mathbb{E}[N_{1:n}]\approx\frac{Td}{n}-\sqrt{\frac{2Td(n-d)\log(n)}{n^2}},\quad\textup{as}\quad  n\rightarrow\infty.
        \label{minApprox}
    \end{equation}{}
However, since this approximations are derived using central limit theorem, negative values are not excluded. In fact, it is easy to verify that for $T<2(n/d-1)\textup{log}n$ (\ref{minApprox}) is negative. Therefore, an approximation of the number of balls in the minimum loaded urn could be given as
    \begin{equation}
        \mathbb{E}[N_{1:n}]\approx\textup{max}\left\{0,\frac{Td}{n}-\sqrt{\frac{2Td(n-d)\textup{log}(n)}{n^2}}\right\},
    \end{equation}{}
as $n\rightarrow\infty$, which completes the proof.
\end{proof}
The approximations for the number of balls in the maximum/minimum loaded urns together with their experimental values are plotted in Fig. \ref{fig:max-min}. Although obtained for large values of $n$, the approximations follow the experimental values very closely for small values of $n$ as well. This conformity holds for number of balls in both the maximum and the minimum loaded urns, and across different values of $d$. In Fig. \ref{fig:lbf}, approximations and experimental values of LBF are plotted. The two values follow a similar trend and their gap continues to shrink as $d$ increases.

    \begin{figure}[t]
        \centering
        \includegraphics[width=.6\columnwidth]{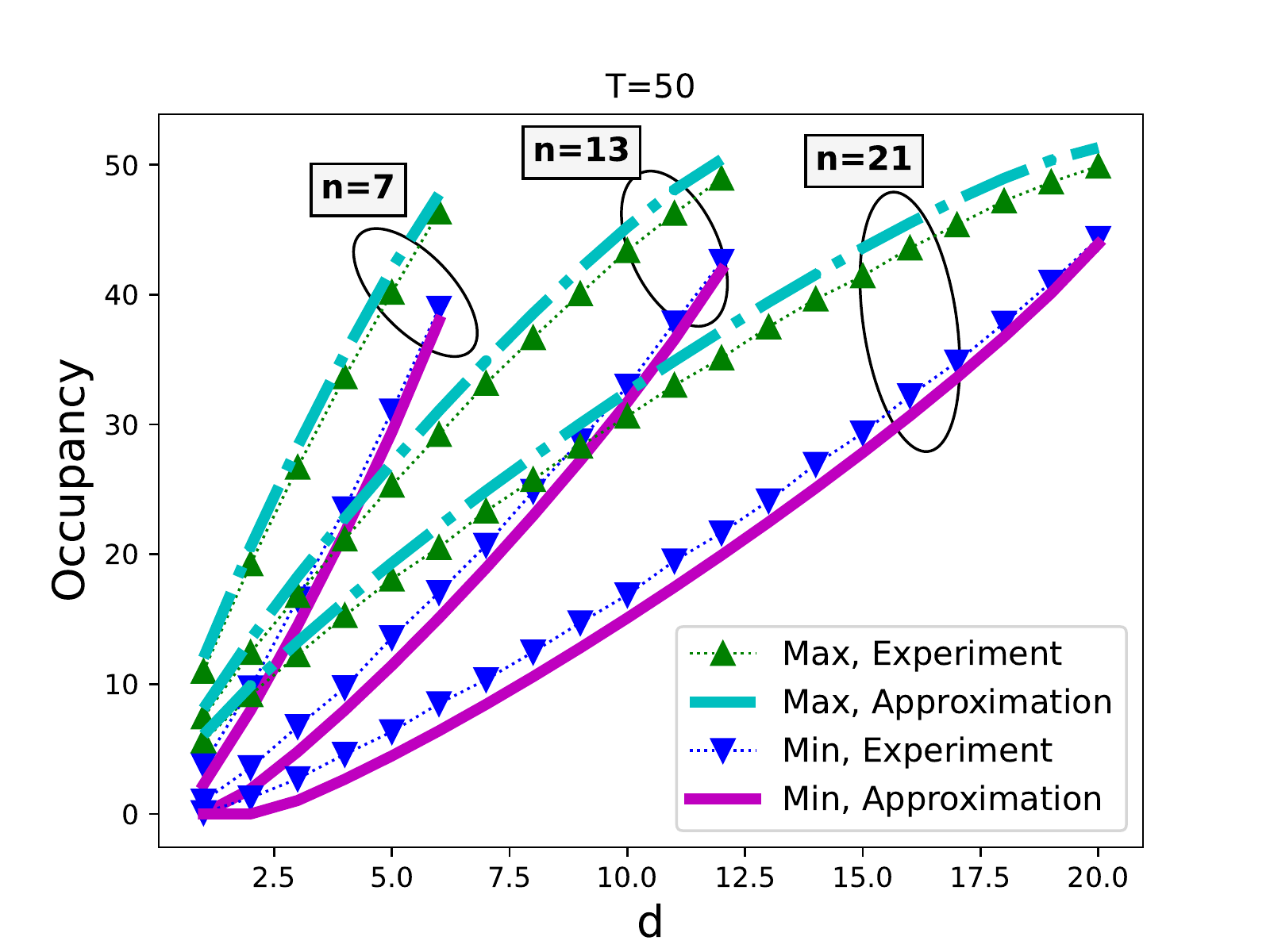}
        \caption{Average number of balls in the maximum and the minimum loaded urns after throwing 50 balls $d$-redundantly versus $d$, for three different $n$. The approximation for both minimum and maximum values follow the experimental values.}
        \label{fig:max-min}
    \end{figure}

\begin{lemma}
After throwing $T$ balls $d$-redundantly to $n$ urns, when urns are chosen uniformly at random for each ball,
    \begin{equation}
        \begin{split}{}
            \textup{AOF}_{random}&=\frac{n}{d^2},\\
            \textup{ODF}_{random}&=\frac{n(n-1)}{d^2(n+d(d-2))}.
        \end{split}
    \label{ODFrandom}
    \end{equation}{}
\end{lemma}{}
\begin{proof}
For two sets of urns of size $d$, chosen uniformly at random and independent from one another, the probability that they overlap in exactly $k$ urns is given by
    \begin{equation}\label{equ:overlapDistRand}
        p_k=\frac{\binom{d}{k}\binom{n-d}{d-k}}{\binom{n}{d}}, \qquad k=0,1,\dots,r.
    \end{equation}{}
Therefore,
    \begin{equation}
        \begin{split}{}
            \mathbbm{E}[X]&=\frac{1}{\binom{n}{d}}\sum_{k=1}^{d}k\binom{d}{k}\binom{n-d}{d-k},\\
                &=\frac{d^2}{n},
                \end{split}
    \label{randomFirstMoment}
    \end{equation}{}
and
    \begin{equation}
        \begin{split}{}
            \mathbbm{E}[X^2]&=\frac{1}{\binom{n}{d}}\sum_{k=1}^{d}k^2\binom{d}{k}\binom{n-d}{d-k},\\
                &=\frac{d^2(n+d(d-2))}{n(n-1)}.
        \end{split}
    \label{randomSecendMoment}
    \end{equation}{}
Respectively, substitution of (\ref{randomFirstMoment}) and (\ref{randomSecendMoment}) in (\ref{equ:AOF}) and (\ref{equ:ODF}) completes the proof.
\end{proof}

It could be seen that both AOF and ODF decrease with $d$. In other words, increasing the number of replicas increases the average the average overlapping urns between the balls and also decreases the evenness of the number of overlapping urns between the balls. On the other hand, increasing $n$ increases both AOF and ODF. That is, for a fixed $d$ increasing $n$ reduces the average number of overlapping urns between the balls and increases the evenness of the number of overlapping urns between the balls.

    \begin{figure}[t]
        \centering
        \includegraphics[width=.6\columnwidth]{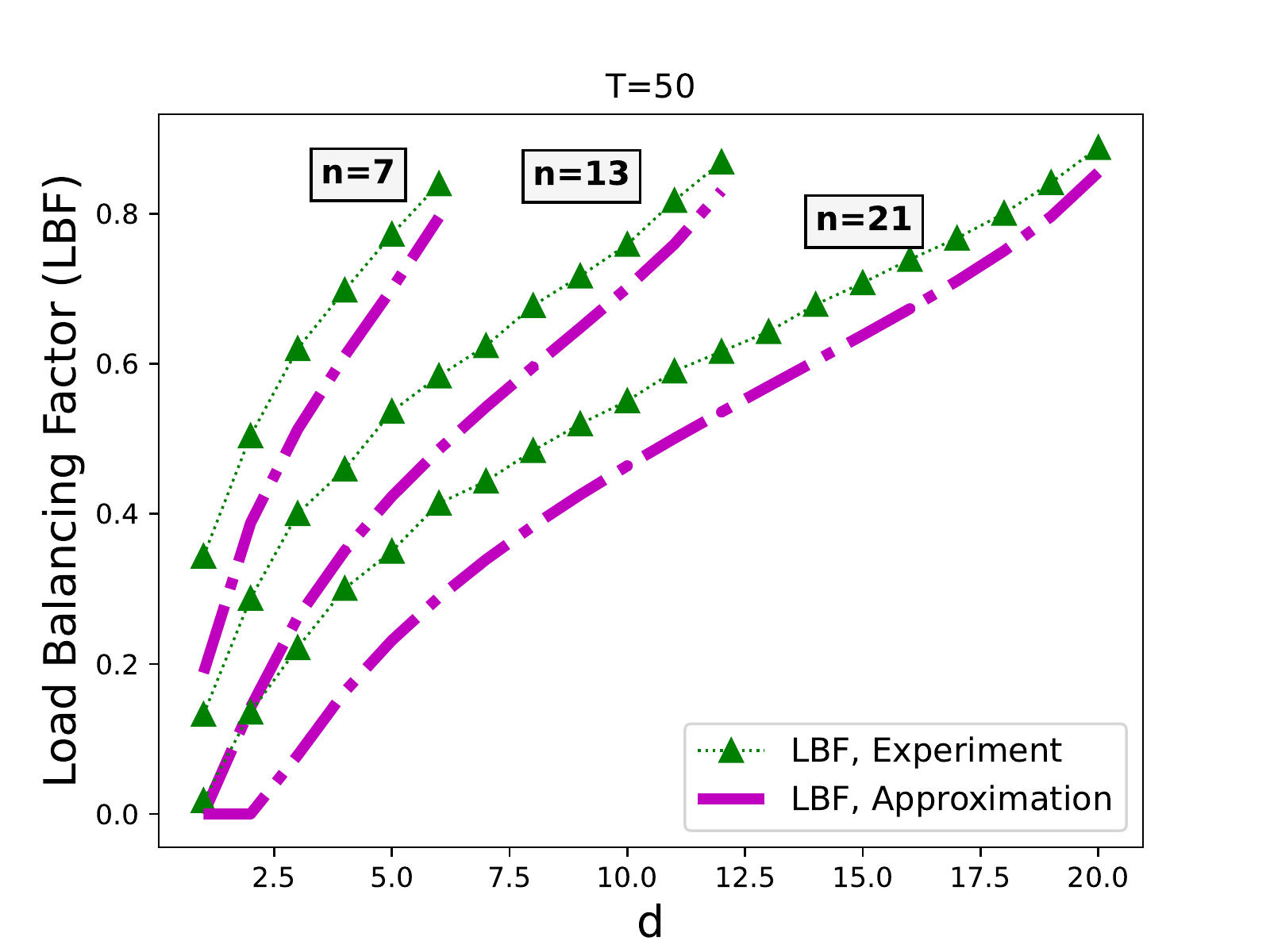}
        \caption{Load balancing factor (LBF) after throwing 50 balls $d$-redundantly versus $d$, for three different $n$. The approximations follows the experimental values. For a given $d$, LBF decreases with $n$.}
        \label{fig:lbf}
    \end{figure}

\subsection{Round-Robin Scheduling}
With round-robin scheduling, a ball is $d$-redundantly thrown into a set of urns that are chosen according to a round-robin policy, as follows. Suppose that the first ball being thrown to the (empty) urns is indexed 1 and each subsequent ball increments this index by one. Then, the ball with index $i\in\{1,2,3,\dots\}$ is $d$-redundantly scheduled to urns $((i-1)d+1)\mod{n},((i-1)d+2)\mod{n},\dots,((i-1)d+d)\mod{n}$. Accordingly, the performance indicators can be derived as follows.

\begin{lemma}
After throwing $T$ balls $d$-redundantly to $n$ urns, when urns are chosen according to the round-robin policy for each ball, $\textup{LBF}_{round-robin}=1$.
\end{lemma}
\begin{proof}
With round-robin scheduling, the difference between the number of balls in the maximum loaded urn and the minimum loaded urn is at most one. Therefore, their ratio is 1 in the limit when $T\to\infty$.
\end{proof}
\begin{lemma}
With round-robin scheduling,
    \begin{equation}
        \begin{split}
            \textup{AOF}_{round-robin}&=\frac{n}{d^2},\\
            \textup{ODF}_{round-robin}&=\frac{3n}{2d^3+d}.
        \end{split}{}
    \end{equation}{}
    \label{ODF:roundRobin}
\end{lemma}{}
\begin{proof}
The first ball, which is thrown to empty urns, is indexed $i=1$. Without loss of generality, consider the set of urns for the ball $i=1$ to be $\{1,2,\dots,d\}$. After throwing $T|n$ balls, there will be $T-1$ sets of urns, chosen according to round-robin policy. Then, it is easy to verify that the number of sets that overlap with the set for the first ball in exactly $k$ urns is $\frac{2T}{n}-1$, for $k\in[1,d-1]$, and $\frac{T}{n}-1$, for $k=d$. Therefore, the probability of a set of urns chosen for a ball indexed $i\geq2$, to have exactly $k$ overlaps with that of $i=1$ is
    \begin{equation}\label{equ:overlapDistRR}
        Pr\{o_X=k\} =
            \begin{cases}
                1-\frac{2d-1}{n}+\frac{d}{T} \quad & k=0,\\
                \frac{2}{n}-\frac{1}{T} \quad & k=1,2,\dots,d-1,\\
                \frac{1}{n}-\frac{1}{T} \quad & k=d,\\
                0 & \textup{otherwise}.
            \end{cases}       
    \end{equation}
Therefore,

    \begin{equation}
        \begin{split}
            \mathbb{E}[o_X]&=\left(\frac{1}{n}-\frac{1}{T}\right)r+\sum_{k=1}^{d-1}k\left(\frac{2}{n}-\frac{1}{T}\right)\\
            &=\left(\frac{1}{n}-\frac{1}{T}\right)r+\left(\frac{2}{n}-\frac{1}{T}\right)\frac{d(d-1)}{2}.\\
        \end{split}{}
        \label{EX_rr}
    \end{equation}
Finally,
    \begin{equation}
        \text{AOF}=\lim_{T\to\infty}\frac{1}{\mathbbm{E}[o_X]}=\frac{n}{d^2}.
    \end{equation}
Furthermore,
    \begin{equation}
        \begin{split}
            \mathbb{E}[o_X^2]&=\left(\frac{1}{n}-\frac{1}{T}\right)d^2+\sum_{k=1}^{d-1}k^2\left(\frac{2}{n}-\frac{1}{T}\right)\\
            &=\left(\frac{1}{n}-\frac{1}{T}\right)d^2+\left(\frac{2}{n}-\frac{1}{T}\right)\frac{2d^3-3d^2+d}{6},
        \label{EX2_rr}
        \end{split}{}
    \end{equation}{}
and
    \begin{equation}
        \text{ODF}=\lim_{T\to\infty}\frac{1}{\mathbbm{E}[o_X^2]}=\frac{3n}{d(2d^2+1)}.
    \end{equation}
\end{proof}

The behaviour of AOF and ODF with respect to $n$ and $d$ for the round-robin policy is the same with the random policy. That is, both metrics increase with $n$ and decrease with  $d$.

\section{Scheduling with Block Designs}\label{sec:BIBD}
 We first provide a basic description of block designs. Then, we introduce a scheduling policy based on a specific design and show that it improves the performance indicators introduced in \ref{subsec:perf_indicators}.

\subsection{Block Designs Basics}
In combinatorics, a block design is a pair $(\CMcal{X},\CMcal{A})$, where $\CMcal{X}$ is a set of objects and $\CMcal{A}$ is a set of non-empty subsets of $\CMcal{X}$, called \textit{blocks} \cite{stinson2007combinatorial}. There exist several combinatorial designs, each satisfying certain balance or symmetry property across the blocks. In this work, we consider one type of design called \textit{Balanced and Incomplete Block Design} (BIBD). In the following, we first define the BIBD and then discuss their benefits for scheduling in a queuing system.

A $(\nu,k,\lambda)$-BIBD is a set $\CMcal{A}$ of subsets of a finite collection $\CMcal{X}$ of objects that satisfies the following conditions,
    \begin{enumerate}
        \item $|\CMcal{X}|=\nu$,
        \item every block consists of $k$ objects, and
        \item every pair of distinct objects is contained in exactly $\lambda$ blocks.
    \end{enumerate}{}
A BIBD is symmetric if the number of objects and the number of blocks are identical. A necessary condition for the existence of a symmetric BIBD is $\lambda(\nu-1)=k(k-1)$. In this work, we consider the symmetric BIBD with $\lambda=1$, and refer to them as BIBD for simplicity. An example of a symmetric BIBD with $\lambda=1$ is shown in Fig.~\ref{fig:FanoScheduling}. The figure shows the $(7,3,1)$-BIBD, where $\CMcal{X}=\{J_1,\dots,J_7\}$ $\CMcal{A}=\{\textcolor{brown}{J_1J_2J_3},\textcolor{red}{J_1J_4J_5},\textcolor{blue}{J_1J_6J_7},\textcolor{teal}{J_2J_5J_7},\textcolor{cyan}{J_3J_4J_7},\textcolor{purple}{J_2J_4J_6},J_3J_5J_6\}$. 

\begin{figure}[hbt]
    \centering
    \includegraphics[scale=1.5]{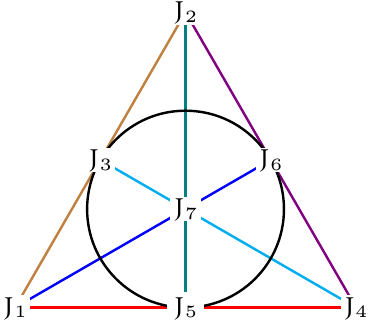}
    \caption{$(7,3,1)$-BIBD where the seven points correspond to seven jobs. Each job is assigned to three servers (lines in the figure). Two different jobs (points) share only a single server (line).}
    \label{fig:FanoScheduling}
\end{figure}

\subsection{Scheduling with BIBD}
Let us consider a queuing system, where each job is $d$-redundantly scheduled to the existing $n$ servers. Suppose a $(n,d,1)$-BIBD exists. Then, the BIBD scheduler is formally defined as follows.

\begin{definition}
In a queuing system with $n$ servers, the BIBD scheduler assigns each job redundantly to $d$ servers according to the objects in a block of a $(n,d,1)$-BIBD, provided that such a design exists.
\end{definition}

In $(n,d,1)$-BIBD, every pair of two blocks have exactly one object in common. Therefore, with the BIBD scheduling, a given job shares all $d$ servers with some jobs, ``$d$-overlap jobs'' , and it shares only one server with other jobs, ``1-overlap jobs''. When blocks are chosen randomly for arriving jobs, then for a given job, the number of $d$-overlap jobs is $O(1/n)$ and the number of 1-overlap jobs is $O(1-1/n)$. Thus, when $n$ is large, the number of ``$d$-overlap jobs'' goes to zero. In other words, the number of overlapping servers between any 2-pairs of jobs is asymptotically equal to one. This is a desirable consequence of the BIBD scheduling; the number of servers a job shares with the rest of the jobs is asymptotically even across the jobs (see discussion in Section \ref{subsec:perf_indicators}). The even overlap between the jobs results in maximum diversity in the queuing time. Accordingly, BIBD scheduling maximizes the benefit of diversity if queuing systems with redundancy.

We choose the parameters $n$ and $d$ such that there exist a $(n,d,1)$-BIBD. A necessary condition for such existence is $n=d(d-1)+1$. In scheduling with BIBD, an incoming job is assigned to the servers indicated by the objects in a block. Blocks can be chosen in a round-robin fashion or randomly across the jobs. In this work, we consider round-robin selection due to the simplicity of its performance indicators. We leave the study of other selection policies as a subject for future studies. In the remainder, BIBD scheduling refers to BIBD scheduling with round-robin selection of blocks, unless stated otherwise. In the following, we derive the performance indicators of BIBD scheduling in the corresponding urns and balls problem.

\begin{lemma}
After throwing $T$ balls $d$-redundantly to $n$ urns, when urns are chosen according to the BIBD scheduling, $\textup{LBF}_{BIBD}=1$.
\end{lemma}
\begin{proof}
With BIBD scheduling, the difference between the number of balls in the maximum loaded urn and the minimum loaded urn is at most one. Therefore, their ratio is 1 in the limit when $T\to\infty$.
\end{proof}

\begin{lemma}
With BIBD scheduling,
    \begin{equation}
        \begin{split}{}
            \textup{AOF}_{BIBD}&=\frac{n}{n+d-1},\\
            \textup{ODF}_{BIBD}&=\frac{n}{n+d^2-1}.
        \end{split}
    \end{equation}{}
\end{lemma}
\begin{proof}
The proof follows the same in principle as the proof of Lemma \ref{ODF:roundRobin}. Except the probability distribution for the random variable $o_X$ is,
    \begin{equation}
         Pr\{o_X=k\} =
            \begin{cases}
                \frac{n-1}{n} \quad & k=1,\\
                \frac{1}{n}-\frac{1}{T} \quad & k=d,\\
                0 & \textup{otherwise}.
            \end{cases}    
    \end{equation}
Therefore,
    \begin{equation}
        \begin{split}{}
            \mathbb{E}[o_X]&=\frac{n-1}{n}+\left(\frac{1}{n}-\frac{1}{T}\right)d\\
        \end{split}
        \label{EX_bibd}
    \end{equation}
and
    \begin{equation}
        \text{AOF}_{BIBD}=\lim_{T\to\infty}\frac{1}{\mathbbm{E}[o_X]}=\frac{n}{n+d-1}
    \end{equation}
Furthermore,
    \begin{equation}
        \begin{split}{}
            \mathbb{E}[X^2]&=\frac{n-1}{n}+\left(\frac{1}{n}-\frac{1}{T}\right)d^2\\
        \end{split}
        \label{EX2_bibd}
    \end{equation}
and
    \begin{equation}
        \text{ODF}_{BIBD}=\lim_{T\to\infty}\frac{1}{\mathbbm{E}[o_X^2]}=\frac{n}{n+d^2-1}
    \end{equation}
\end{proof}
\section{Numerical Comparison of Performance Indicators}\label{sec:numerical_comp}
In this section, we provide a numerical comparison of the performance indicators of the three scheduling policies. To compare the scheduling policies, we consider values for $n$ and $d$ that satisfy the necessary condition for the existence of symmetric BIBD, i.e. $n=d(d-1)+1$.
    \begin{figure}[t]
        \centering
        \includegraphics[width=.6\columnwidth]{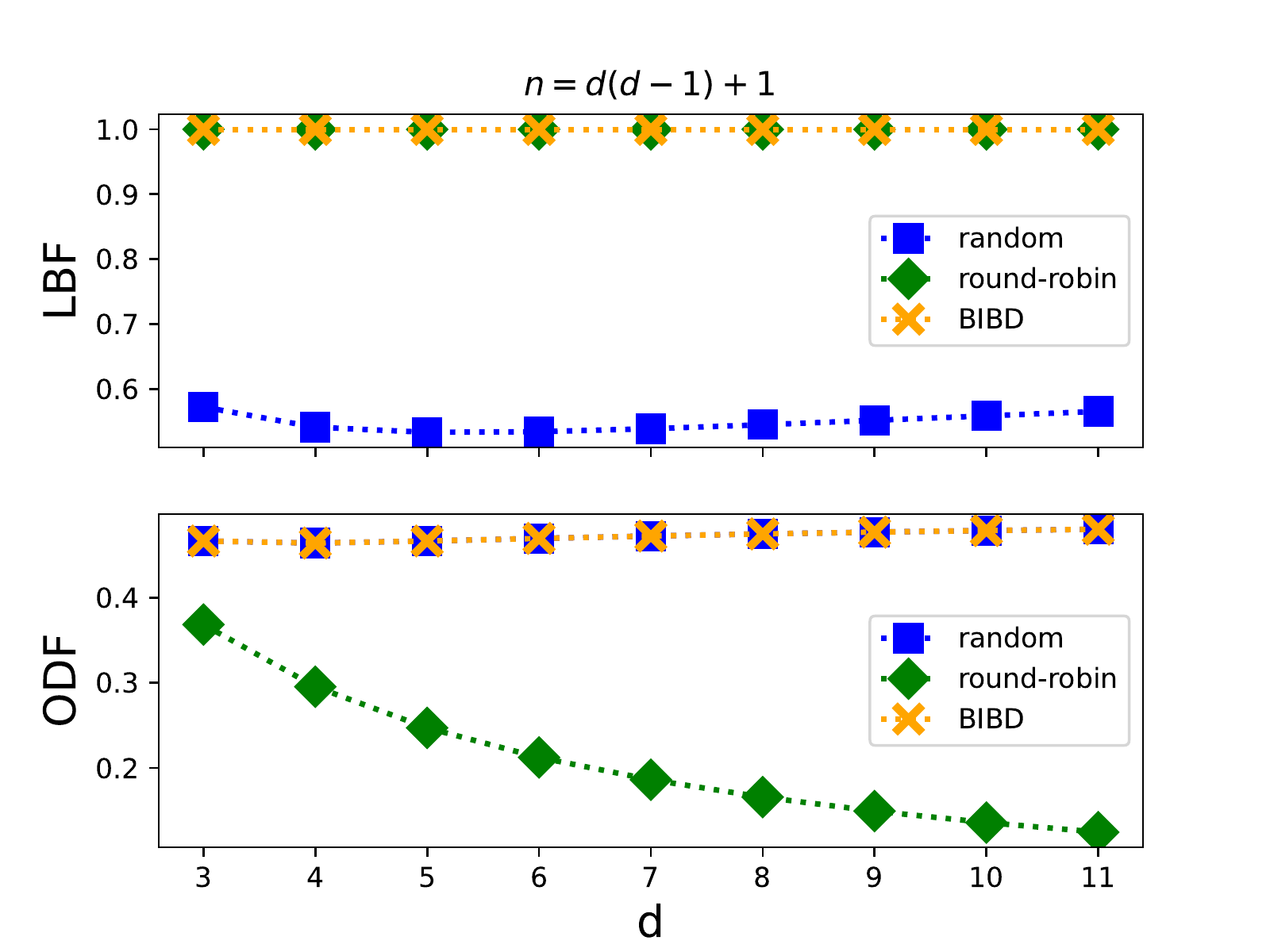}
        \caption{LBF and ODF for random, round-robin and BIBD policies versus $d$, when $n=d(d-1)+1$. Round-robin policy has the largest LBF and the smallest ODF. On the other hand, random policy has the largest ODF and the smallest LBF. BIBD policy has the best of both worlds; maximum LBF and ODF.}
        \label{fig:lbfrdf}
    \end{figure}

\begin{table*}[t]
    \caption{Performance indicators of scheduling policies} 
    \centering 
    {\large
    \begin{tabular}{l c c c} 
        \hline\hline  
        Policy & LBF($\leq 1$) & AOF($\leq 1$) & ODF($\leq 1$)
        \\ [0.5ex]
        \hline\hline 
        Random &  $\textup{max}\left\{0,\frac{\frac{Td}{(d-1)^2+d}-\sqrt{\frac{2Td(d-1)^2\textup{log}((d-1)^2+d)}{((d-1)^2+d)^2}}}{\frac{Td}{(d-1)^2+d}+\sqrt{\frac{2Td(d-1)^2\textup{log}((d-1)^2+d)}{((d-1)^2+d)^2}}}\right\}$ & $\frac{(d-1)^2+d}{d^2}$ & $\frac{(d-1)^2+d}{d(2d-1)}$\\
        Round-robin & 1 & $\frac{(d-1)^2+d}{d^2}$& $\frac{3[(d-1)^2+d]}{d(2d^2+1)}$\\
        BIBD & 1 & $\frac{(d-1)^2+d}{d^2}$ & $\frac{(d-1)^2+d}{d(2d-1)}$\\
        \hline 
    \end{tabular}}
    \label{tab:indicators}
\end{table*}

The performance indicators of the three scheduling policies are provided in Table~\ref{tab:indicators} and are plotted in Fig.~ \ref{fig:lbfrdf}. The load balancing indicator LBF is 1 for round-robin and BIBD policies. These two policies can provide a perfect balance of load across the servers by assigning an equal number of jobs to each server. However, when the service requirements of jobs are diverse, e.g., following a bimodal distribution, the number of jobs is not an effective measure of the actual workload on each server. Thus, the performance benefit of such load balancing schemes vanishes as service requirements get skewed across the jobs. The LBF of the random policy is upper bounded by 1. For a fixed $n$, one can expect the LBF of random policy to increase with $d$. However, with $n=d(d-1)+1$, it stays almost constant as $d$ changes.

When $n$ and $d$ satisfy the necessary condition of the existence of BIBD, the three scheduling policies have identical AOF. The same average overlap in round-robin and BIBD policies is intuitive. That is because, with the same number of balls thrown, the number of balls in the urns is (almost) identical across the policies, which implies that the AOF should also be identical. However, a less intuitive result is the equality of AOF of the random policy with round-robin and BIBD. With random scheduling, the number of balls may not be equal across the urns. Thus, the balls in the maximum loaded urn would have higher overlap (than the average overlap), and those in the minimum loaded bin would have lower overlap. Nevertheless, the average number of overlaps across the balls is the same with round-robin and BIBD, where urns would have (almost) an equal number of balls.

With BIBD, the number of overlapping urns across the balls is either one or $d$. In contrast, with the random policy, the number of overlapping urns across the balls could be any integer in $[0,d]$. Nevertheless, the diversity in the number of overlaps across the balls $\mathbbm{E}[o_X^2]$ is equal for the random and BIBD policies; see ODF plot in Fig.~\ref{fig:lbfrdf}. On the other hand, the round-robin policy results in higher diversity in the number of overlaps, i.e., lower ODF. Even though the number of overlapping urns in the round-robin could be any integer in $[0,d]$ (similar to the random policy), its ODF is lower than the ODF of the random policy. That is because the number of overlaps has different distributions in the two policies, see (\ref{equ:overlapDistRand}) and (\ref{equ:overlapDistRR}).

BIBD scheduling has the best of both worlds; it has the highest LBF (together with the round-robin) and the highest ODF (together with the random). Accordingly, our performance indicators suggest that the BIBD scheduling policy is the best for redundancy queuing systems.

\section{Graph-Based Analysis of Scheduling Policies}\label{sec:graph}
In this section, we connect the overlapping properties of the scheduling policies to the expansion properties of their associated graphs. In particular, we show that a graph constructed by the BIBD policy is a better expander than a graph constructed by the round-robin policy. 

\subsection{Incident Structure of Scheduling Policy}
Let us consider a queuing system with $n$ servers where jobs are assigned to the servers $d$-redundantly. Servers are chosen according to a block design $\CMcal{P}$. Suppose the block design has $n$ blocks, each with $d$ objects. We define the incident structure for such a design as an $n\times n$ binary matrix $\CMcal{M}_{n,d}^\CMcal{P}$, where the entry at row $i$ and column $j$ is $1$ if object $i$ belongs to the block $j$, and is $0$ otherwise. The value of an entry is specified by $\CMcal{P}$. The sum of entries in each row of $\CMcal{M}_{n,d}^{\CMcal{P}}$ is $d$. As an example, the incident structure for the round-robin policy with $n=4$ and $d=2$ is given by,
    \begin{equation}
            \CMcal{M}_{4,2}^{\text{round-robin}}=
            \left[ {\begin{array}{cccc}
            1 & 1 & 0 & 0 \\
            0 & 1 & 1 & 0 \\
            0 & 0 & 1 & 1 \\
            1 & 0 & 0 & 1 \\
            \end{array} } \right].
        \label{rrExample}
    \end{equation}{}
An incident structure $\CMcal{M}_{n,d}^{\CMcal{P}}$ can be interpreted by a left $d$-regular bipartite graph $(V_{n,d}^l\cup V_n^r,E^\CMcal{P})$, such that $V_{n,d}^l$ are the vertices in the left subgraph, $V_n^r$ are the vertices in the right subgraph, $|V_{n,d}^l|=|V_n^r|=n$. Furthermore, $E^\CMcal{P}$ is the set of edges induced by the policy $\CMcal{P}$. As an example, the graph associated with a (7,3,1)-BIBD is given in Fig. \ref{fig:bipartite}.

    \begin{figure}[t]
        \centering
        \includegraphics[width=.5\columnwidth]{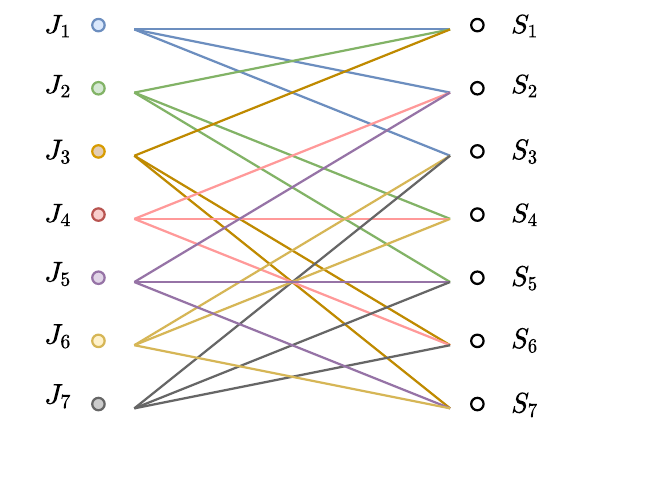}
        \caption{Bipartite graph associated with the incident structure of a (7,3,1)-BIBD. Every pair of two jobs (colors) appear on exactly one server.}
        \label{fig:bipartite}
    \end{figure}

\subsection{Spectral Expansion of Incident Structure}
It is well known that the expansion properties of an incident structure are closely related to the spectral properties of its adjacency matrix. We consider incident structures which are described by $d$-regular bipartite graphs. In particular, a regular bipartite graph is an expander if and only if the second largest eigenvalue of its adjacency matrix is well separated from the first \cite{alon1986eigenvalues}. Next, we formally define an expander structure based on the eigenvalues of its adjacency matrix.

\begin{definition}
Let $\CMcal{G}=(V,E)$ be a $d$-regular graph, where $|V|=n$. Let $\lambda_1\geq\lambda_2\geq\dots\geq\lambda_n$ be the eigenvalues of the adjacency matrix of $\CMcal{G}$.
The spectral gap of the graph $\CMcal{G}$ is defined as $\Delta(\CMcal{G})=d-\max\{\lambda_2,|\lambda_n|\}$.
\end{definition}{}
In general, larger spectral gap corresponds to better expansion. 
When $\CMcal{G}$ is d-regular and \textit{bipartite}, the eigenvalues of the adjacency matrix of $\CMcal{G}$ are symmetric around 0. In other words, the eigenvalues of the adjacency matrix are either 0 or pairs of $\lambda$ and $-\lambda$ \cite{brito2018spectral}. Thus, $\max\{\lambda_2,\lambda_n\}$ has the trivial value of $d$. In such cases, the spectral gap is defined as follows.

\begin{definition}\cite{brito2018spectral}
Let $\CMcal{G}=(V,E)$ be an $d$-regular bipartite graph, where $|V|=n$. Let $\lambda_1\geq\lambda_2\geq\dots\geq-\lambda_2\geq-\lambda_1$ be the eigenvalues of the adjacency matrix of $\CMcal{G}$.
The spectral gap of the graph $\CMcal{G}$ is defined as $\Delta(\CMcal{G})=d-\lambda_2$.
\end{definition}

The associated graphs with round-robin and BIBD scheduling policies are $d$-regular bipartite, with spectral gap $d-\lambda_2$. In the following, we derive the spectral gap of these two policies.

\subsection{Round-robin Structure}
Given the incident matrix $\CMcal{M}_{n,d}^{\text{round-robin}}$, the adjacency matrix is given by
    \begin{equation}
        A^{\text{round-robin}}=\left[\begin{array}{cc}
            \bf{0}_{n\times n} & \CMcal{M}_{n,d}^{\text{round-robin}} \\
             \text{tr}(\CMcal{M}_{n,d}^{\text{round-robin}}) & \bf{0}_{n\times n}
        \end{array}{}
        \right],
    \label{rrAdjacency}
    \end{equation}
where $\text{tr}(\cdot)$ is the transpose operator. The following theorem characterizes the spectral gap of the round-robin structure based on the eigenvalues of $A^{\text{round-robin}}$.
\begin{theorem}
The spectral gap $\epsilon^{\text{round-robin}}$ of the round-robin structure with $d$-regular bipartite incident graph $\CMcal{M}_{n,d}^{\textup{round-robin}}$ is given by,
    \begin{equation}
        \epsilon^{\text{round-robin}}=d-\frac{\sin{d\pi/n}}{\sin{\pi/n}}.
    \label{rrGap}
    \end{equation}
\label{theorem2}
\end{theorem}
\begin{proof}
See Appendix.
\end{proof}

\begin{corollary}
The spectral gap of the round-robin structure is zero in the limit when $n\rightarrow\infty$.
\end{corollary}
According to the behaviour of the spectral gap, the round-robin structure is not asymptotically an expander. As it will be shown later in this section, this structure is not an expander even when $n=d(d-1)+1$. This result has been also confirmed by the performance metric ODF (see Fig.~\ref{fig:lbfrdf}), as it decreases as $d$ grows large. Therefore, round-robin policy provides more even overlaps when $d$ is small. 

\subsection{BIBD Structure}

It is known that BIBDs have largest spectral gap among all d-regular bipartite graphs and thus are optimal expanders \cite{hoholdt2009optimal}. We restate the following theorem from \cite{hoholdt2009optimal}.
\begin{theorem}
The spectral gap of an $(n,d,1)-BIBD$ construction is
    \begin{equation}
        \epsilon^{\text{BIBD}}=d-\sqrt{d-1}.
    \label{bibdGap}
    \end{equation}{}
\end{theorem}
The following theorem gives the asymptotic behavior of the spectral gap of round-robin and BIBD structures.
\begin{theorem}
The spectral gap of the round-robin incident structure, with parameters $n$ and $d$ such that $n=d(d-1)+1$, is asymptotically zero in the limit $d\rightarrow\infty$, while the spectral gap of BIBD incident structure grows indefinitely in the limit as $d\rightarrow\infty$.
\end{theorem}
\begin{proof}
From (\ref{rrGap}),
    \begin{equation*}
        \begin{split}
            \lim_{d\rightarrow\infty} \epsilon^{\text{round-robin}}&=\lim_{d\rightarrow\infty}d-\frac{\sin{d\pi/n}}{\sin{\pi/n}}\\
            &=\lim_{d\rightarrow\infty}d-\frac{\sin{d\pi/[d(d-1)+1]}}{\sin{\pi/[d(d-1)+1]}}\\
            &=\lim_{d\rightarrow\infty}d-d\\
            &=0.
        \end{split} 
    \end{equation*}
And the indefinite growth of the spectral gap of BIBD incident structure is clear from (\ref{bibdGap}).
\end{proof}

    \begin{figure}[t]
        \centering
        \includegraphics[width=.6\columnwidth,keepaspectratio]{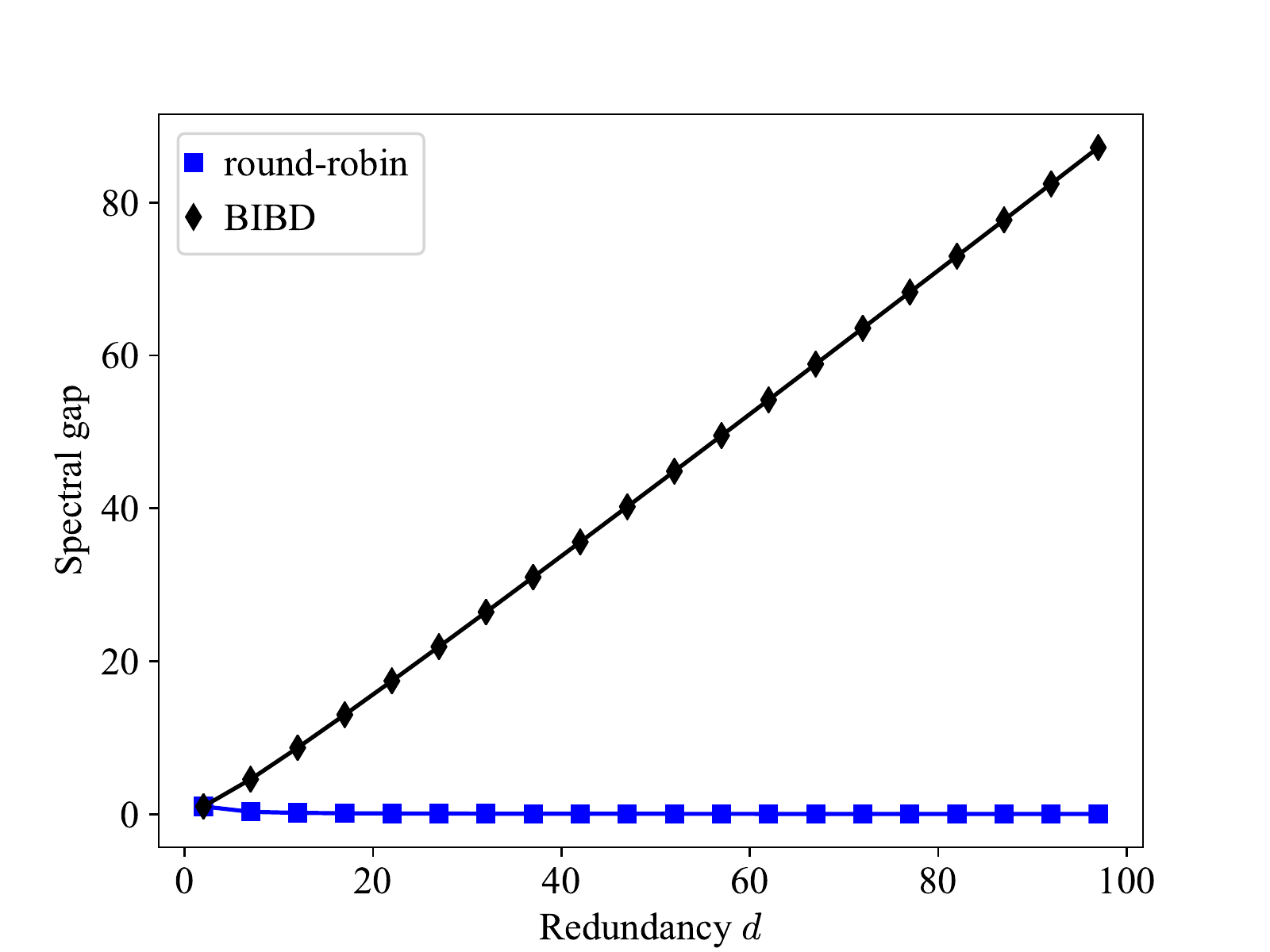}
        \caption{Spectral gap of round-robin and BIBD policies. The spectral gap of BIBD increases monotonically with $d$. The spectral gap of round-robin goes to zero with $d$.}
        \label{fig:specral_gap}
    \end{figure}

From (\ref{rrGap}), as $d$ gets larger, the spectral gap of round-robin construction will shrink. This is inline with our performance indicators, in that ODF decreases with $d$.

\section{Simulation Results}\label{sec:simResult}
In this section we provide simulations results for the average queuing time with the random, round-robin and BIBD scheduling policies. In the simulated environment, there are $n$ servers and a job is scheduled to $d$ different servers, selected according to a given scheduling policy. The redundant copies of a job get cancelled once the first copy starts the service. We assume job cancellations are instantaneous. If there is an idle server among the selected servers, an arriving job would enter the service instantaneously, with no more redundant copies. In that case, if there are more than one idle server, ties are broken arbitrarily.

The jobs' service times are sampled from a 2-phase hyper-exponential distribution. That is, the service time of an arriving job is sampled from a slow exponential distribution, i.e., the job requires a long service, with probability $p$ and from a fast exponential distribution, i.e. the job requires short service, with probability of $1-p$. The probability distribution of a 2-phase hyper-exponential distribution is given by,
    \begin{equation*}
        f_S(s)=(1-p)f_{S_1}(s)+pf_{S_2}(s),
    \end{equation*}{}
where,
    \begin{equation*}
        f_{S_i}(s) =
            \begin{cases}
                \mu_ie^{-\mu_is} &\quad s\geq 0,\\
                0 & s<0,
            \end{cases}
    \end{equation*}{}
and $p$ is the probability that a job requires long service. This behaviour of jobs' service requirements is observed in practice, e.g., Google Traces \cite{chen2010analysis}.

We denote by $q$ the ratio of the average service time of a long job $1/\mu_2$ to that of a short job $1/\mu_1$. It can be verified that the average service requirement on the system is $\mu=(1-p(1-1/q))$. For brevity, we assume $\mu_1=1$. Job arrival follows a Poisson process with inter-arrival rate of $\lambda$. We further denote by $\rho=\lambda/\mu$ the average load on the system. We provide simulation results for $\rho\in[0,1)$, i.e., the stability region of the queuing system.

Each figure in this section consists of two subfigures. The left sub-figure shows the results for the low to medium load regime and the right sub-figure shows the results for medium to high load regime. The set of parameters for each plot is given by the 4-tuple, $(n,d,q,p)$.

Fig. \ref{fig:13_4_10_10} shows the average queuing times, with $(n,d,q,p)=(13,4,10,0.1)$. In all arrival rates, BIBD policy outperforms both random and round-robin. The improvement ranges from about $10\%$ in the high load regime up to $20\%$ in lower to medium loads. With $(n,d,q,p)=(21,5,10,0.1)$, Fig. \ref{fig:21_5_10_10} shows that the relative performance of random and BIBD policies is almost the same as that in Fig. \ref{fig:13_4_10_10}. However, the gap between the round-robin and BIBD increases. This is consistent with the behaviour of ODF in Fig. \ref{fig:lbfrdf} in that it reduces with $d$. This behaviour could be further justified by the fact that large values of $d$ results in more overlap in the system. Therefore, with larger $d$ a less overlap-aware policy would perform poorer. In Fig. \ref{fig:13_4_10_10}, the BIBD policy reduces the queuing time by $10\%$ and $25\%$ compared to random and round-robin, respectively.

    \begin{figure}[t]
        \centering
        \includegraphics[width=.8\columnwidth,trim={1.8cm 0 2.7cm 0},clip]{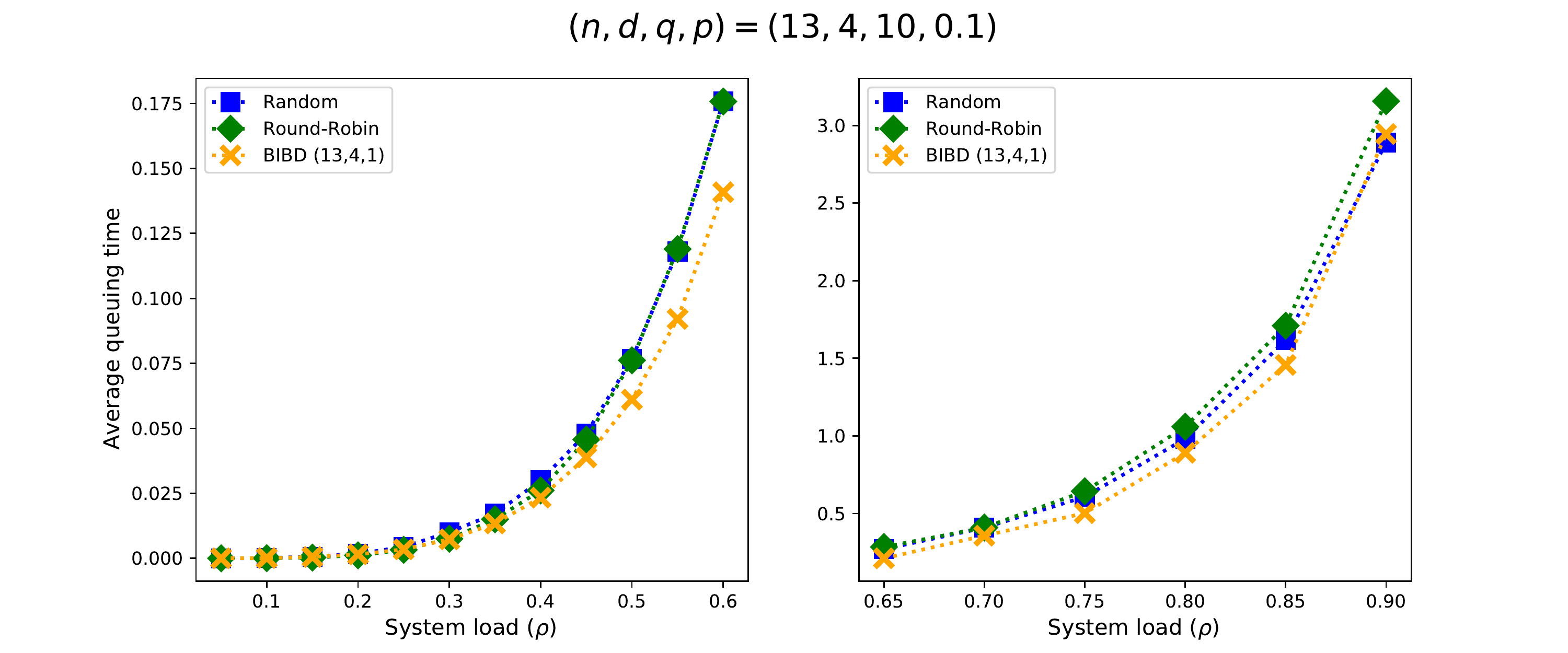}
        \caption{Average queuing time for nonadaptive scheduling policies, with $(n,d,q,p)=(13,4,10,0.1)$. }
        \label{fig:13_4_10_10}
    \end{figure}
    
    \begin{figure}[t]
        \centering
        \includegraphics[width=.8\columnwidth,trim={1.8cm 0 2.7cm 0},clip]{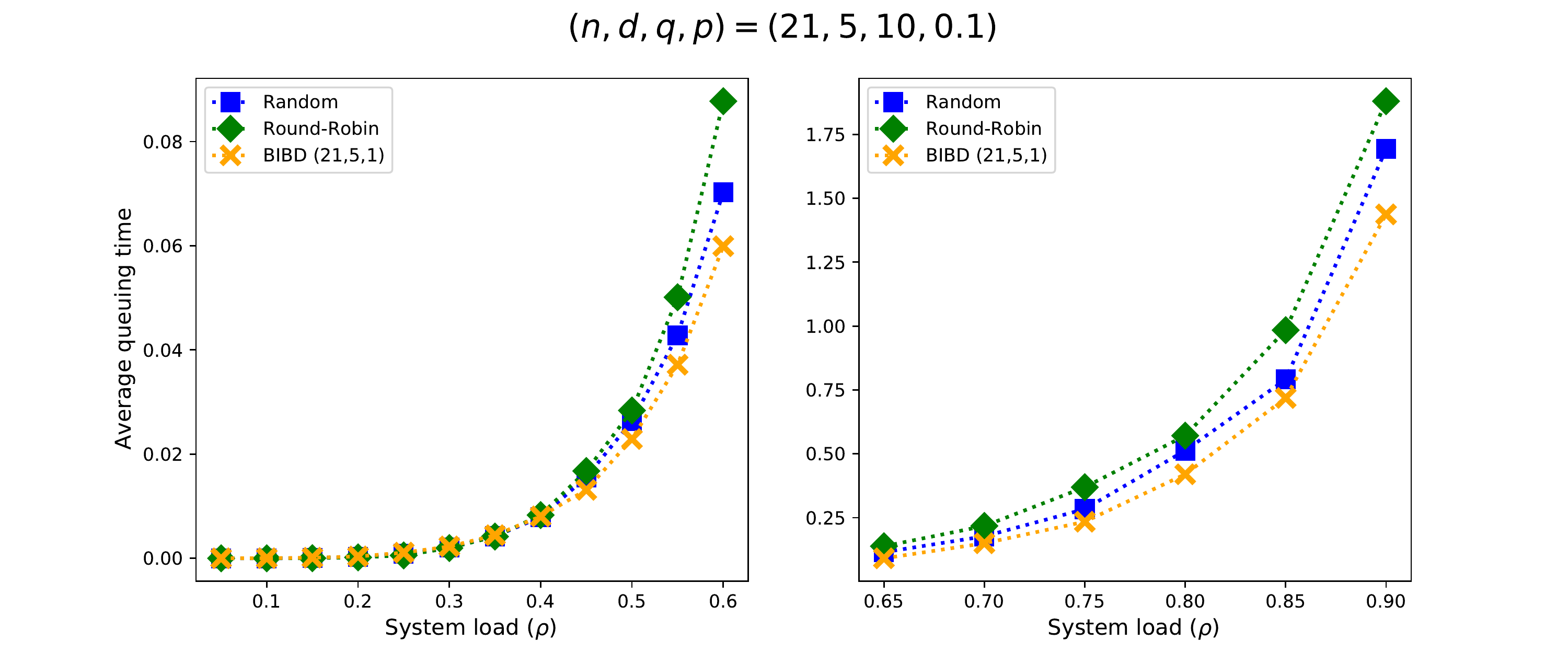}
        \caption{Average queuing time for nonadaptive scheduling policies, with $(n,d,q,p)=(21,5,10,0.1)$. }
        \label{fig:21_5_10_10}
    \end{figure}

Next, we keep $n$, $d$ and $p$ fixed and only increase $q$. We would like to observe the effect of more diverse service requirements on the performance of the scheduling policies. In Fig. \ref{fig:21_5_50_10}, we set $(n,d,q,p)=(21,5,50,0.1)$. It can be seen that the performance gap between round-robin and BIBD/random increases. This observation holds for both low and high traffic regimes. With this set of parameters, BIBD reduces the queuing time by $100\%$ when compared to round-robin policy. It is worth mentioning that, although round-robin is an effective policy for load balancing, when the probability of long jobs is small and their average service requirement is large, load balancing without considering the service requirement of jobs is not effective. However, when the probability of large jobs is high, Fig. \ref{fig:21_5_50_50} with $(n,r,q,p)=(21,5,50,0.5)$, round-robin performs closer to BIBD. Furthermore, the random policy, which is more effective in handling overlaps, has a closer performance to BIBD when $p$ is smaller. This happens because, with small $p$ long jobs are not frequent and even BIBD, which balances the average load, fails to provide load balancing. Therefore, the load balancing capability of a scheduling policy has small impact on its performance when $p$ is small. In Fig. \ref{fig:21_5_50_50}, BIBD outperforms random and round-robin policies by up to $25\%$ and $50\%$, respectively.
    \begin{figure}[t]
        \centering
        \includegraphics[width=.8\columnwidth,trim={1.8cm 0 2.7cm 0},clip]{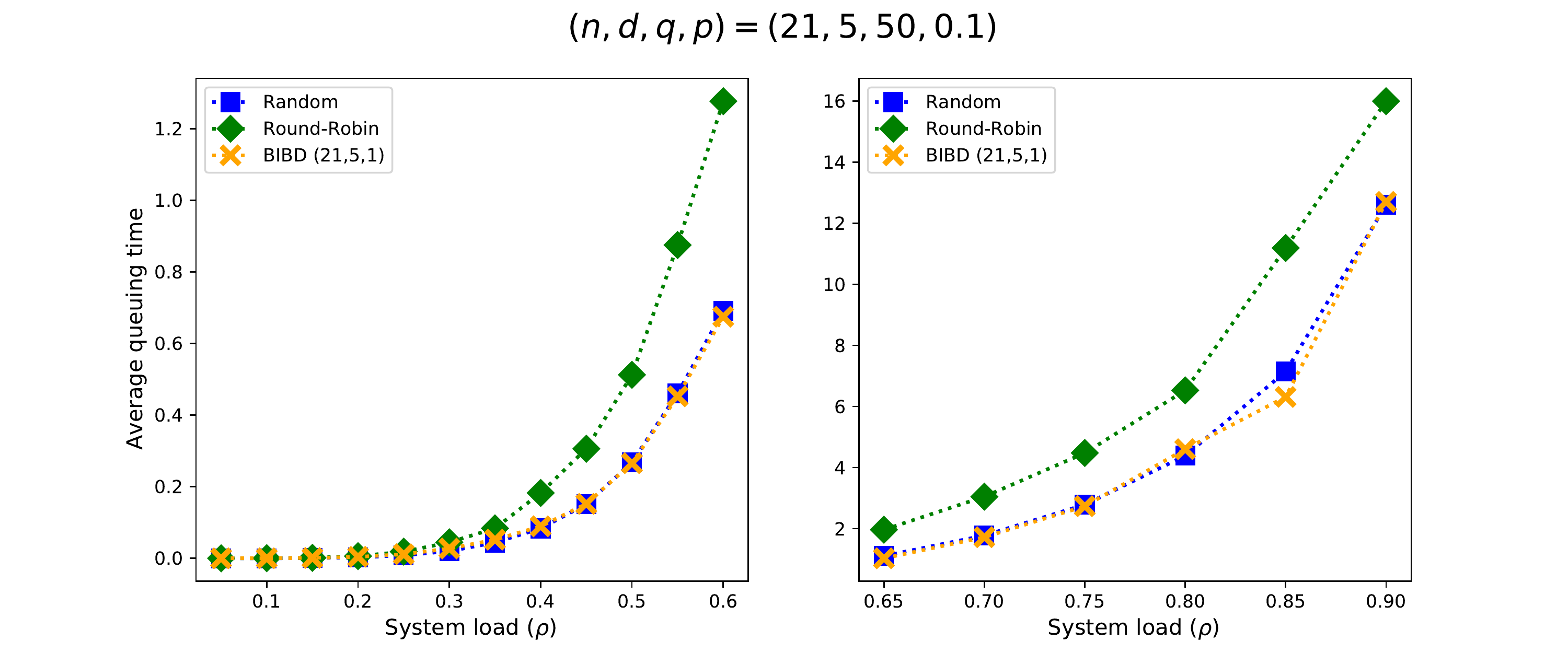}
        \caption{Average queuing time for nonadaptive scheduling policies, with $(n,d,q,p)=(21,5,50,0.1)$. }
        \label{fig:21_5_50_10}
    \end{figure}
    
    \begin{figure}[t]
        \centering
        \includegraphics[width=.8\columnwidth,trim={1.8cm 0 2.7cm 0},clip]{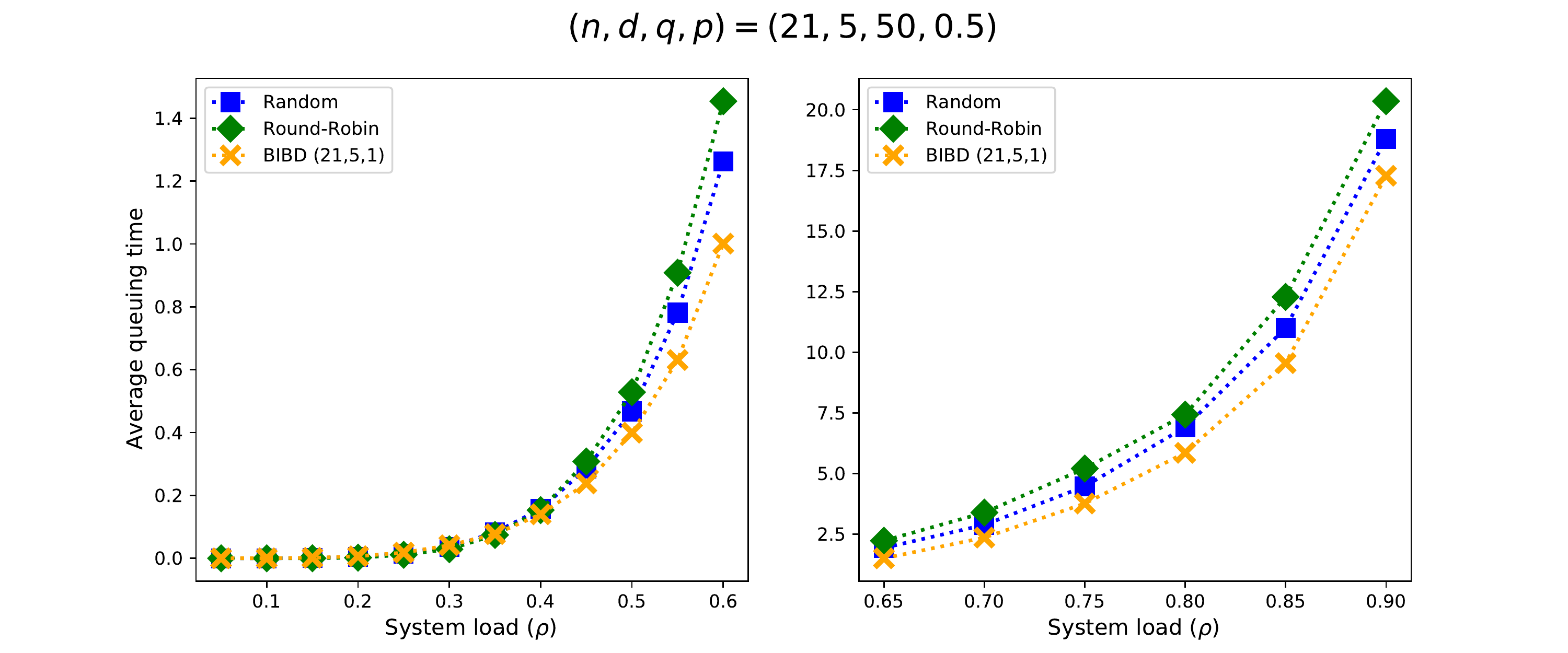}
        \caption{Average queuing time for nonadaptive scheduling policies, with $(n,d,q,p)=(21,5,50,0.5)$. }
        \label{fig:21_5_50_50}
    \end{figure}

\section{Conclusion and Future Directions}\label{sec:conclusion}
We studied scheduling policies that nonadaptively replicate each job to queues in $d$ out of $n$ servers. Only the replica that first reaches the head of the queue gets executed. We argued that the overlaps between job replicas on the servers are good indicators of these systems' queueing time performance. We first derived performance indicators based on the statistics of the overlaps for two classical policies: random and round-robin. We then proposed a nonadaptive scheduling policy based on a combinatorial block design. The performance indicators of the block design-based policy were superior to those of the random and round-robin policies. We further investigated the pattern of overlapping servers for the round-robin and block design-based policies by studying the expansion properties of their associated graphs. It turns out that the better the expansion property of a scheduling policy's graph, the better the performance of the policy in queuing systems with redundancy. Finally, by simulating the studied queuing system, we observed that the proposed scheduling policy achieves a lower average queuing time than the random and round-robin policies.

This work can be extended in multiple directions. The most interesting and maybe the hardest problem would be to derive a closed form for the average performance of block design-based scheduling policy, and to compare it with the performance of the random policy derived in \cite{gardner2017redundancy}. Another direction may consider designing scheduling policies based on other combinatorial designs. The dependencies between the parameters of a combinatorial design could impose restrictions on the system parameters, and lifting such restrictions could also be a topic for future research.

\newpage
\bibliographystyle{IEEEtran}
\bibliography{ref}

\newpage

\appendix

\subsection*{Proof of Theorem \ref{theorem2}:}

1) We first observe that eigenvalues of round-robin incident structure with parameters $n$ and $d$, can be written as \cite{gray2006toeplitz},
    \begin{equation}
        \phi_m=\sum_{l=0}^{d-1}e^{-j2\pi ml/n}, \quad m=0,1,\dots,n-1.
    \label{rrEigenValues}
    \end{equation}{}
    
2) We next use the following lemmas to connect the eigenvalues of the incident structure to the eigenvalues of its adjacency matrix.


\begin{lemma}
Round-robin policy is \ul{normal}, that is, its incident structure satisfies $\CMcal{M}\CMcal{M}^{T}=\CMcal{M}^{T}\CMcal{M}$.
\label{lem3}
\end{lemma}

Since the round-robin policy is normal, the singular values of its incident structure are the absolute values of its eigenvalues. Lemma \ref{lem4} follows from \eqref{rrEigenValues} and Lemma \ref{lem3}.

\begin{lemma}
The singular values of the incident structure of round-robin policy, with parameters $n$ and $d$, are given by
    \begin{equation}
        \psi_m=\left|\sum_{l=0}^{d-1}e^{-j2\pi ml/n}\right|, \quad m=0,1,\dots,n-1.
    \end{equation}{}
\label{lem4}
\end{lemma}

It is well known that the eigenvalues of (\ref{rrAdjacency}) are symmetric with respect to zero, and that the positive eigenvalues are the singular values of $\CMcal{M}_{\text{round-robin}}$. Accordingly, Lemma \ref{lem5} follows.

\begin{lemma}
The second largest eigenvalue of matrix (\ref{rrAdjacency}) is given by,
    \begin{equation}
    \begin{split}
        \psi_1&=\left|\sum_{l=0}^{d-1}e^{-j2\pi l/n}\right|\\
            &=\frac{\sin{d\pi/n}}{\sin{\pi/n}}.
    \end{split}{}
    \end{equation}{}
\label{lem5}
\end{lemma}{}

3) Finally, we conclude that the largest eigenvalue of (\ref{rrAdjacency}) is $d$, the spectral gap of the round-robin policy is 
    \begin{equation*}
        \epsilon^{round-robin}=d-\frac{\sin{d\pi/n}}{\sin{\pi/n}}.
    \end{equation*}{}

\end{document}